\documentclass[pra, superscriptaddress, reprint, aps, a4paper, showkeys]{revtex4-1}

\usepackage[pdfauthor={Sebastian Meznaric}]{hyperref}
\usepackage{graphicx}
\usepackage{amsmath}
\usepackage{amsthm} 
\usepackage{amssymb}
\usepackage{amsbsy}
\usepackage{fullpage}
\usepackage{natbib}
\usepackage[usenames]{color}
\usepackage{mathrsfs}

\newcommand{\ket}[1]{\left|#1\right\rangle}
\newcommand{\bra}[1]{\left\langle #1 \right|}

\newcommand{\braopket}[3]{\left\langle #1 \hspace{1mm} \vline \hspace{1mm} #2 \hspace{1mm} \vline \hspace{1mm} #3 \right\rangle}

\newcommand{\uket}[1]{|#1\rangle}
\newcommand{\ubra}[1]{\langle #1 |}

\newcommand{\mket}[1]{ \left \lVert #1 \right \rangle}
\newcommand{\mbra}[1]{\left\langle #1 \right \rVert}

\newcommand{\mbraopket}[3]{\left\langle #1 \hspace{1mm} \vline \hspace{0.5mm} \vline \hspace{1mm} #2 \hspace{1mm} \vline \hspace{0.5mm} \vline \hspace{1mm} #3 \right\rangle}

\newcommand{\op}[1]{\mathbf{#1}}

\definecolor{DarkRed}{rgb}{0.7,0,0}

\newtheorem{theorem}{Theorem}

\DeclareMathOperator{\Tr}{Tr}

\begin{document}

\title{Introduction to Measurement Space and Application to Operationally Useful Entanglement and Mode Entanglement}
\author{Sebastian Meznaric}
\email{s.meznaric1@physics.ox.ac.uk}
\affiliation{Clarendon Laboratory, University of Oxford, Oxford OX1 3PU, United Kingdom}

\date{\today}

\begin{abstract}


We introduce a concept of the measurement space where the information that is not accessible using the particular type of measurements available is erased from the system. Each state from the Hilbert space is thus mapped to its counterpart in the measurement space. We then proceed to compute the entanglement of formation on this new space. We find that for local measurements this never exceeds the entanglement of formation computed on the original state. Finally we proceed to apply the concept to quantum communication protocols where we find that the success probability of the protocol using the state with erased information in combination with perfect measurements is the same as using the non-perfect measurements and the original state. We thus postulate that the so defined entanglement measure quantifies the amount of useful entanglement for quantum communication protocols. 

\end{abstract}


\maketitle

\section{Introduction}

In an experimental setting it is likely never the case that we can achieve perfect rank-1 projective measurements. In such a setting it is impossible for us to obtain all the possible information about the state of the system. In this paper we construct a theoretical framework that erases those properties of the state that cannot affect the outcome of our (imperfect) measurements.


Mathematically, we define a new Hilbert space, called the measurement space. The states in the measurement space will contain only the information that can be extracted using the available measurements. We construct these states through a procedure not unlike the Naimark's dilation theorem, whereby we change the state so that the measurements act just like rank-1 projectors in the measurement space. Intuitively, the properties that are preserved are those that are now accessible only with the perfect measurements (rank-1 projectors).

Mathematically, the map can best be understood as taking all states in the Hilbert space into the form $\sum_j \sqrt{p_j} \mket{\psi_j}$, where $\mket{\psi_j}$ are orthogonal measurement space states corresponding to different measurement outcomes. The double bar in the notation is there to remind us that these states are the measurement space versions thereof. When we erased the non-measurable properties we changed the components of the state that correspond to the measurement outcomes - they were made orthogonal (see figure \ref{fig:MeasurementSpace}). Since the perfect rank-1 projection measurements are best we can get, we effectively transform the state into the form where the best possible measurements cannot extract more information than the chosen generalized measurements in the original Hilbert space. 

\begin{figure*}
 \centering
 \includegraphics[scale=0.4]{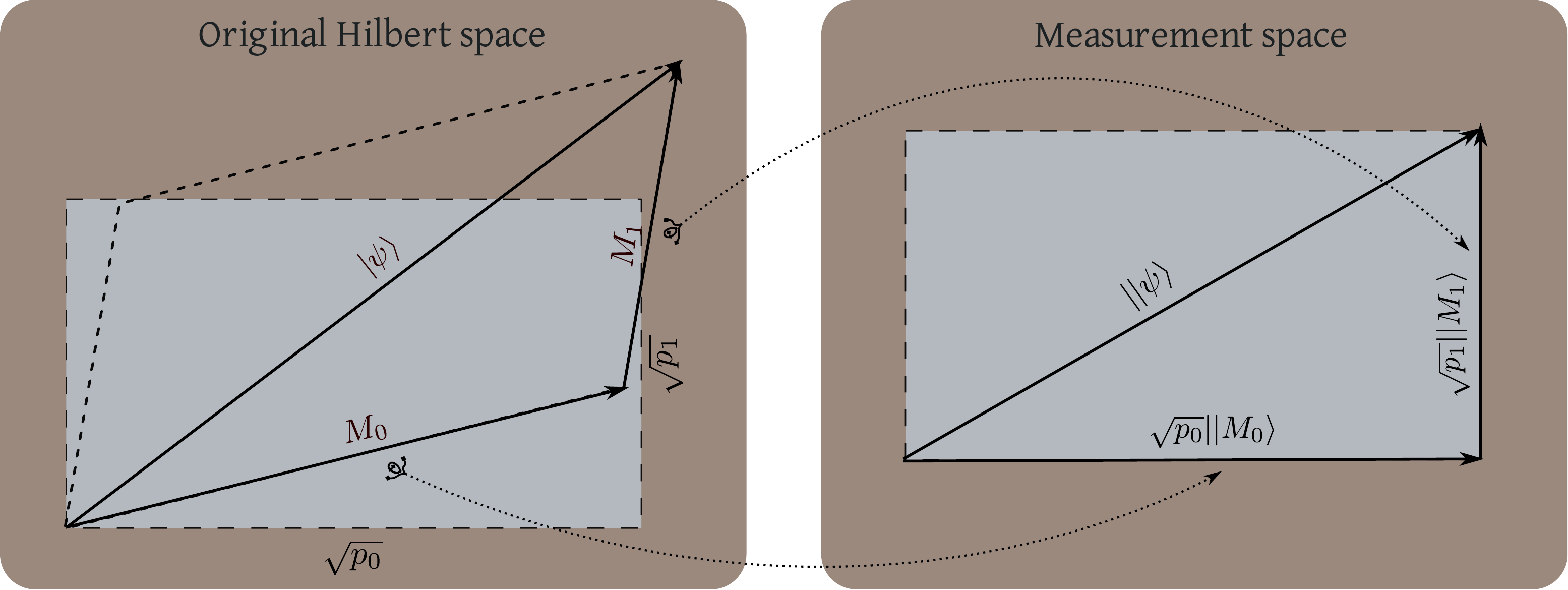}
 \caption{When a state $\ket{\psi}$ is mapped to the measurement space the components corresponding to the measurement outcomes are made orthogonal and their lengths adjusted to the square roots of the respective outcome probabilities.}
 \label{fig:MeasurementSpace}
\end{figure*}

Physically, we can look at the first stage of the measurement process - the interaction between the measurement device and the measured system. We add the measurement apparatus ancilla in some state $\ket{0}$ to the state $\ket{\psi}$ to obtain $\ket{\psi}\ket{0}$. The following unitary $U$ then describes the interaction (for more details see for instance \cite{NielsenChuang, VonNeumann32}):
\begin{align}
 U \ket{\psi} \otimes \ket{0} = \sum_m M_m \ket{\psi} \otimes \ket{m}. \label{eq:VonNeumann}
\end{align}
 The read out of the measurement is now conducted by conducting a measurement on the measurement device. The measurement operators are of the form $\op{1}\otimes\ket{m}\bra{m}$. Notice that if you write $M_m \ket{\psi} \otimes \ket{m}$ as a single state vector, you get exactly the measurement space state. The measurable information in the state (\ref{eq:VonNeumann}) is the same as that in the measurement space.

Next we explore how the quantum information theoretic quantities behave on states where the non-measurable information has been removed. In particular we are interested in entanglement and show that the state without any non-measurable information is just as effective in all quantum communication protocols as the original state. More precisely, using a state with this amount of entanglement and perfect measurement operators for a quantum communication protocol results in the same fidelity as using the imperfect measurements and the original state. 

For some protocols the fidelity is a strictly monotone increasing function of entanglement when perfect rank-1 projectors are used. In such cases we can claim that the entanglement in the measurement space state is the amount of entanglement that is useful for the particular protocol. The operational entanglement measure thus obtained encompasses those devised previously to deal with the indistinguishable particles (like for example in \cite{Tichy09, Wiseman03}) since the measurement operators can always be made insensitive to the exchange of particles. 



We also prove that when the measurements are local the map from the original state to the measurement space is LOCC. This implies that the entanglement in the measurement space is always less than or equal to the standard entanglement. Furthermore, we find an upper bound for the two-dimensional systems and show when the upper bound is attained.

Finally, we apply the above formalism to the mode entaglement. We find that the isolated single-particle mode entanglement cannot have an operational meaning unless further degrees of freedom are added to the system. We examine the more general case of adding more particles and more degrees of freedom and find the upper bound on the operational entanglement. Whenever the total number of measurement outcomes is not a prime number mode entanglement is found to have an operational meaning as well. 

Following this introduction, the paper is organised into the definition of measurement space in section \ref{sec:Definition}, the application to entanglement in section \ref{sec:OperationalEntanglement}, followed by a look at the mode entanglement in section \ref{sec:ModeEntanglement} and concluding remarks in section \ref{sec:Conclusion}.

\section{Definition of measurement space}\label{sec:Definition}

Here we shall consider the quantum system $S$ to consist of a state $\ket{\psi}$ and a set of generalized measurement operators. We shall denote the original Hilbert space as $\mathcal{H}$ and we define a new Hilbert space, called the measurement space, and denote it as $\mathcal{M}$. Every state $\ket{\psi} \in \mathcal{H}$ has a corresponding measurement space state denoted as $\mket{\phi}$. The map from $\ket{\psi}$ to $\mket{\psi}$ is defined as
\begin{align}
 \mket{\psi} = m_S(\ket{\psi}) = \sum_m \left(\braopket{\psi}{M_m^\dagger M_m}{\psi}\right)^{1/2} \mket{m} \label{eq:MeasState},
\end{align}
where $\mket{m}$ are the orthonormal basis states of the measurement space $\mathcal{M}$ and $M_m$ are the generalized measurements. The dimension of $\mathcal{M}$ is therefore the same as the number of measurement outcomes in our quantum system $S$. As argued in the introduction, the measurement space state contains exactly the properties of the quantum reality measurable with the particular set of measurements. 

The above formulation is a type of nonlinear Naimark dilation, where we pay the price of nonlinearity in order to have simplicity in the calculation of the map. 


\section{Entanglement in measurement space}\label{sec:OperationalEntanglement}

In this section we shall examine the properties exhibited by the entanglement in the measurement space. In other words, we shall consider the following quantity
\begin{align}
 E_m(\ket{\psi}) = E(\mket{\psi}),
\end{align}
where $E$ is an entanglement measure. 

We know that every protocol can be implemented by starting with a measurement by Alice, classical communication of the measurement result to Bob, followed by a unitary by Bob (for proof of this see \cite{NielsenChuang}). We shall label the measurement operators that Alice uses as $M_1, ..., M_n$. In response to Alice obtaining the result $k$, Bob implements the unitary $U_k$. In order to measure the fidelity of the protocol, Bob devises two measurements: $M_{y,k}$ corresponding to the protocol succeeding and $M_{n,k}$ corresponds to protocol failing. As before, the index $k$ enumerates Alice's measurement outcomes and is there for the cases where the success and failure of the protocol depends on Alice's measurement result. In many practical cases this is so. For example, the quantum key distribution protocol succeeds if Bob obtains the same measurement result as Alice when they measure in the same basis.


We are now ready to prove our theorem.
\begin{theorem}
 Given a state $\ket{\psi}$ to be used as a resource in a quantum communication protocol, the success rate of the protocol is the same with the original state and original \emph{imperfect} measurements as it is with the measurement space state and rank-1 projective measurements.
\end{theorem}
\begin{proof}
 Denote $p_{k,y} = \bra{\psi} M_k^\dagger M_k \otimes M_{y,k}^\dagger M_{y,k} \ket{\psi}$ the probability of the outcome corresponding to Alice obtaining $k$ and Bob finding that the protocol succeeded, while $p_{k,n}$ is defined similarly. First we shall find the success probabilities in the measurement space. The measurement space state is given by
 \begin{align}
  \mket{\psi} = \sum_k \left( \sqrt{p_{k,y}} \mket{k} \mket{y,k} + \sqrt{p_{k,n}} \mket{k} \mket{n,k} \right).
 \end{align}
 The success probability given that Alice obtained the outcome $k=m$ is given by
 \begin{align}
  p(y \, | \, m) & = \frac{p(y \cap m)}{p(m)} = \frac{p_{m,y}}{\mbraopket{\psi}{ \mket{m}\mbra{m} \otimes \op{1}}{\psi}} \\
  & = \frac{p_{m,y}}{p_{m,y}+p_{m,n}}.
 \end{align}
 The overall probability of success is then given by $\sum_m p(y\, | \, m) p(m) = \sum_m p_{y,m}$.
 
 Now we repeat the same calculation in the original Hilbert space, where we find
 \begin{align}
  p(m) & = \braopket{\psi}{M_m^\dagger M_m \otimes \op{1}}{\psi} \\
  & = \bra{\psi}\left(M_m^\dagger M_m \otimes M_{m,y}^\dagger M_{m,y} \right. \\ + & \left. M_m^\dagger M_m \otimes M_{m,n}^\dagger M_{m,n}  \right)\ket{\psi} \\
  & = p_{m,y} + p_{m,n}.
 \end{align}
 So as above, we obtain 
 \begin{align}
  p(y \, | \, m) = \frac{p_{m,y}}{p_{m,y}+p_{m,n}}.
 \end{align} This completes the proof.
\end{proof}

%
%
%
Next we shall show that as long as the measurement operators are local, entanglement in the measurement space is always less than that in the original Hilbert space. If the measurements are not local, this can easily be shown not to be the case with an example. 
\begin{theorem}
 Given a set of local generalized measurement operators $M_{m_A}^A \otimes M_{m_B}^B$ and a state $\ket{\psi}$, it is always true that $E(\mket{\psi})\leq E(\ket{\psi})$.
\end{theorem}
\begin{proof}
The existence of an LOCC map taking $\ket{\psi}$ to $\mket{\psi}$ would immediately imply that $E(\mket{\psi}) \leq E(\ket{\psi})$ (see \cite{Nielsen99, NielsenChuang, NielsenNotes2002}).

The state $\ket{\psi}$ is distributed among Alice and Bob. To prove the theorem, we add one measurement apparatus ancilla to each of the parties and initialise them in some state $\ket{0}$ and continue according to the measurement scheme:
\begin{align}
 U \ket{\psi} \ket{0} \ket{0} = \sum_{m_A, m_B} M_{m_A}^A \otimes M_{m_B}^B \ket{\psi} \ket{m_A} \ket{m_B}.
\end{align}
This operation is local and unitary and so $E(\ket{\psi}\ket{0}\ket{0}) = E(U \ket{\psi} \ket{0} \ket{0})$. Now we construct a projection operator for Alice of the form
\begin{align}
 P^A = \sum_{m_A} P_{m_A}^{j_{m_A}} \otimes \ket{m_A}\bra{m_A},
\end{align}
where $P_{m_A}$ is a rank-1 projector corresponding to some measurement outcome. The operators must be so constructed that each outcome $j_{m_A}$ has the probability of $\frac{1}{n_A}$, where $n_A$ is the dimension of Alice's original Hilbert space. This is achieved by making the operators project onto the quantum fourier transforms of the eigenstates of Alice's partial density matrix appearing in front of $\ket{m_A}\bra{m_A}$. The partial density matrix describing Alice's subsystem of $U \ket{\psi} \ket{0} \ket{0}$ is given by  
\begin{widetext}
 \begin{align}
  \Tr_B \left[U \ket{\psi} \ket{0} \ket{0} \right] = \sum_{m_A} \sum_{m'_A}\sum_{m_B} \Tr_B\left[ M_{m_A}^A \otimes M_{m_B}^B \ket{\psi}\bra{\psi} {M_{m'_A}^A}^\dagger \otimes {M_{m_B}^B}^\dagger \right] \otimes \ket{m_A}\bra{m'_A}.
 \end{align}
 \end{widetext}  Now consider only the part $\sum_{m_B} \Tr_B\left[ M_{m_A}^A \otimes M_{m_B}^B \ket{\psi}\bra{\psi} {M_{m_A}^A}^\dagger \otimes {M_{m_B}^B}^\dagger \right]$ (notice that prime disappeared over the second $m_A$). This is a non-normalized partial density matrix and thus has a set of orthogonal eigenstates. Label these states as $\uket{\phi_{m_A}^1}, \uket{\phi_{m_A}^2}, \ldots, \uket{\phi_{m_A}^{n_A}}$ ($n_A$ here is the dimension of Alice's ancilla). Now do a quantum Fourier transform on these states and label the results as $\uket{\omega_{m_A}^1}, \uket{\omega_{m_A}^2}, \ldots, \uket{\omega_{m_A}^{n_A}}$. I.e.,
\begin{align}
 \uket{\omega_{m_A}^j} = \frac{1}{\sqrt{n_A}}\sum_{k=0}^{n_A - 1} e^{2 \pi i j k / n_A} \uket{\phi_{m_A}^k}. \label{eq:MeasureFourier}
\end{align}
 It is clear that the probability of measuring $\uket{\omega_{m_A}^j}$ is equal to $\frac{1}{n_A}$ for any $m_A$. So let us denote $P_{m_A}^{j_{m_A}} = \uket{\omega_{m_A}^{j_{m_A}}}\ubra{\omega_{m_A}^{j_{m_A}}}$. For every $m_A$ we have a number of measurement results $1, \ldots, n_A$. For each instance of measurement, we can then write down a list of results, one corresponding to each $m_A$ and put them in order into a vector $\vec{j}$. 
This is followed by a unitary operation of the form
\begin{align}
 U_A = \sum_{m_A} U_{m_A} \otimes \ket{m_A}\bra{m_A},
\end{align}
where $U_{m_A}$ maps the new state in front of $\ket{m_A}$ to some fixed state $\ket{0}$. Repeating the same for Bob leaves us with the state
\begin{align}
 \ket{0}\otimes \ket{0} \sum_{m_A, m_B} \left(p_{m_A, m_B}\right)^{1/2} \ket{m_A} \otimes \ket{m_B},
\end{align}
where the ancilla is now exactly our measurement space state!
\end{proof}
It is a corollarly of the above theorem that for any separable state $\ket{\psi}$, operationally useful entanglement vanishes.

Konrad et. al. show in \cite{Konrad08} that given a two-partite system of qubits and a LOCC map of the form $\mathcal{L}\otimes\op{1}$ it holds that
\begin{align}
 C\left(\mathcal{L}\otimes\op{1} \ket{\psi}\right) = C\left(\mathcal{L}\otimes\op{1} \ket{\phi^+}\right) C\left(\ket{\psi}\right),
\end{align}
where $C$ is concurrence, a simple to calculate measure of two qubit entanglement \cite{Wooters98}. They also show that in general for LOCC maps that are not of the form $\mathcal{L}\otimes\op{1}$ we have
\begin{align}
 C\left(\mathcal{L} \ket{\psi}\right) \leq C\left(\mathcal{L} \ket{\phi^+}\right) C\left(\ket{\psi}\right).
\end{align}
The discussion below thus applies only to qubits.

We showed above that the transformation from $\ket{\psi}$ to $\mket{\psi}$ is an LOCC map. We shall denote the Alice's part of the map as $\mathcal{L}_A(\ket{\psi})$ and Bob's part of the map as $\mathcal{L}_B(\ket{\psi})$, so that $\mathcal{L}(\ket{\psi}) = \mathcal{L}_A(\ket{\psi}) \otimes \mathcal{L}_B(\ket{\psi})$. Using the above results, we can then say that
\begin{align}
 C\left(\mathcal{L}_A(\ket{\psi}) \ket{\psi}\right) = C\left(\mathcal{L}_A(\ket{\psi})\ket{\phi^+}\right) C\left(\ket{\psi}\right).
\end{align}
But we notice that in the above proof that the only place where our map depends on the state is when we measure in the Fourier transformed basis - equation (\ref{eq:MeasureFourier}). Further, notice that if the operators $M_{m_A}$ are linear combinations of the projectors on the Schmidt basis of the state $\ket{\psi}$, then the measurements in equation (\ref{eq:MeasureFourier}) will simply be the Fourier transforms of the Schmidt basis (to see this consider the eigenstates of $\sum_{m_B} \Tr_B\left[ M_{m_A}^A \otimes M_{m_B}^B \ket{\psi}\bra{\psi} {M_{m_A}^A}^\dagger \otimes {M_{m_B}^B}^\dagger \right]$). Since for the Bell state $\ket{\phi^+}$ any basis is a Schmidt basis, we can thus see that $\mathcal{L_A}(\ket{\phi^+})$ acts in exactly the same way as $\mathcal{L_A}(\ket{\psi})$. Therefore
\begin{align}
 C\left(\mathcal{L}_A(\ket{\psi}) \ket{\psi}\right) = C\left(\mathcal{L}_A(\ket{\phi^+})\ket{\phi^+}\right) C\left(\ket{\psi}\right).
\end{align}
Notice that the right hand side is proportional to $C\left(\ket{\psi}\right)$. 

When the measurements of $B$ happen to also be along the Schmidt basis of $\mathcal{L}_A \ket{\psi}$, we can go further and say
\begin{align}
 C\left(\mathcal{L}_A \otimes \mathcal{L}_B \ket{\psi}\right)  = C\left(\mathcal{L}_A \ket{\phi^+}\right) C\left(\mathcal{L}_B \ket{\psi}\right) \nonumber \\
  = C\left(\mathcal{L}_A \ket{\phi^+}\right) C\left(\mathcal{L}_B  \ket{\phi^+}\right) C\left(\ket{\psi}\right).
\end{align}
Here we omitted the explicit dependence of the maps on the states, since the state they act on and the dependence are the same.

In general, however, we will only have that
\begin{align}
 C\left(\mathcal{L}_A \otimes \mathcal{L}_B \ket{\psi}\right) & \leq \nonumber \\ & C\left(\mathcal{L}_A \ket{\phi^+}\right) C\left(\mathcal{L}_B  \ket{\phi^+}\right) C\left(\ket{\psi}\right).
\end{align}
To see this, notice that if we do not Fourier transform the correct basis to get to equation (\ref{eq:MeasureFourier}), various measurement outcomes in the next stage of the map will no longer occur with equal probabilities. This essentially introduces an additional local positive semi-definite map on one of the subsystems and results in lower entanglement. Indeed, the output of this procedure is generally a mixed state.

What the above means is that for a set of measurements there always exists a unitary that maximizes the amount of measurement space entanglement. The unitary is such that Alice's operators are aligned with the Schmidt basis of the state $\ket{\psi}$, while Bob's are aligned with $\mathcal{L}_A \ket{\psi}$. The amount of entanglement in the measurement space is then a linear function of the amount of entanglement in the original Hilbert space. 

Because of all these results, the entanglement in the measurement space can be considered as being the amount of entanglement that can be operationally used in the context of quantum communication protocols.



\section{Mode Entanglement}\label{sec:ModeEntanglement}

Entanglement of quantum field modes \cite{Heaney06} occurs when dealing with a state of the form 
\begin{align}
 \ket{\psi} = \frac{1}{\sqrt{2}}\left(\op{a}^\dagger +\op{b}^\dagger\right) \ket{0}.
\end{align}
Here $\ket{0}$ is a vacuum state. As an example, consider a single, perfectly isolated spin $1/2$ particle in the state
\begin{align}
 \ket{\psi} = \frac{1}{\sqrt{2}} \left( \ket{0} + \ket{1} \right).
\end{align} Considering the state $\ket{0}$ to be mode $a$ and the state $\ket{1}$ to be mode $b$, we can also write it in the form $\ket{\psi} = \frac{1}{\sqrt{2}} \left(\ket{1_a 0_b} + \ket{0_a 1_b}\right)$, or in the second quantized form as above. This looks just like one of the maximally entangled Bell states. The big question then is: \emph{Can this entanglement be used to facilitate information theoretic protocols?} Previous attempts to answer this question show that this appears to be possible with an addition of a very large particle reservoir \cite{Heaney09-1, Heaney09}. Here we look at this question from a more general perspective.

The mathematical definition of entanglement would have us conclude that the above state is entangled. In fact, any state in a Hilbert space that is at least two dimensional can be written as an entangled state since we can always find two modes such that the probability of the system being in one or the other of the modes is non-zero.

Consider the measurement space map of $\ket{\psi}$. Clearly, there are two possible measurement outcomes for \emph{any} state; these are $\ket{1_a 0_b}\bra{1_a 0_b}$ and $\ket{0_a 1_b}\bra{0_a 1_b}$. We therefore know that $\ket{1_a 0_b}\bra{1_a 0_b} + \ket{0_a 1_b}\bra{0_a 1_b} = \op{1}$ and we know that our set of measurements is complete. Thus, as we have two measurements, the measurement space will also be two dimensional, the basis consisting of $\mket{a}$ and $\mket{b}$, one for the particle being in mode $a$ and the other for $b$. The map of the above state into the measurement space is $\mket{\psi} = \frac{1}{\sqrt{2}}\left(\mket{a} + \mket{b}\right)$. Since 2 is a prime number there is no way to separate the Hilbert space into a tensor product space and thus it is not possible for the \emph{single} particle to posses any measurement space entanglement. However, if we find a way to add more measurements, by for example introducing additional particles or degrees of freedom, the mode entanglement may well acquire measurable consequences. Indeed, it was found that the single-particle mode entanglement can be used in quantum teleportation, as long as the single particle is augmented with a particle reservoir in the form of a Bose-Einstein condensate. In this scenario, the states $\ket{0_a 0_b}$ and $\ket{1_a 1_b}$ effectively become possible to rotate into and so it is no longer true, as it was above, that $\ket{1_a 0_b}\bra{1_a 0_b} + \ket{0_a 1_b}\bra{0_a 1_b} = \op{1}$.

If we simply add additional particles to the system, without chaging their indistinguishable nature, it is also possible to increase the amount of entanglement in the measurement space. We define a separate Fock space mode for each of the possible internal states, in our case $\ket{0}$ and $\ket{1}$. As such we notice that the dimensionality of the resulting Fock space is equal to the number of possible ways the number of particles can be partitioned into two non-negative integers. We consider $n+m$ and $m+n$ as two distinct partitions. For example, when we have two particles with two internal degrees of freedom, our space consists of $\ket{20}$, $\ket{11}$ and $\ket{02}$, where we have that $2 = 2+0 = 1+1 = 0+2$. With $n$ particles our space is then $n+1$ dimensional. If our measurements consist only of particle-number measurements, we can only have operationally useful entanglement when $n+1$ is not a prime number. More generally, assume we have $n$ particles, each of which has $m$ possible internal states and each of the states corresponds to a mode. Then the number of possible measurement outcomes is equal to the number of partitions of $n$ into $m$ integers. In other words, it is the number of ways that $n$ can be written as $n=n_1+n_2+...+n_m$, where $n_j$ are non-negative integers and where a different order of summands counts as a different partition. In general, the number of such partitions is equal to the binomial coefficient $\left( \begin{array}{c} n+m-1 \\ m-1 \end{array} \right)$. Out of all pairs of divisors of the number of partitions, suppose we each time choose the largest one and out of the set of these largest ones we choose the smallest one and call it $p$. Mathematically,
\begin{multline}
 p = \inf\left\{ \max(k,l): k,l \in \mathbb{Z^+} \hspace{2mm} \mathrm{ and } \right. \\ \left. \hspace{2mm} k \cdot l = \left( \begin{array}{c} n+m-1 \\ m-1 \end{array} \right) \right\}.
\end{multline}
Then the bipartite useful entanglement based on the entropy of entanglement is bounded above by
\begin{align}
E_u \leq \log\left[\frac{1}{p} \left( \begin{array}{c} n+m-1 \\ m-1 \end{array} \right) \right].
\end{align}
Notice that the above definition implies that $E_u = 0$ if the binomial coefficient is a prime number and that $p$ will be smallest (and hence the entanglement bound largest) when the binomial coefficient is a perfect square, as expected.  Here we assumed that our measurement set consists only of number-of-particles measurements and that our system is perfectly isolated. Also notice that since $p=2$ is the smallest prime number, it is always true that the following \emph{weaker} inequaltiy holds
\begin{align}
 E_u \leq \log\left[\frac{1}{2} \left( \begin{array}{c} n+m-1 \\ m-1 \end{array} \right) \right].
\end{align}


\section{Conclusion}\label{sec:Conclusion}

We introduced the concept of a measurement space, a map of the Hilbert space where states contain only the information extractable through available measurements. We considered how entanglement measures behave in the measurement space and proved several results, summarised as follows:
\begin{enumerate}
 \item As long as the measurement operators are local, the amount of entanglement in the measurement space never exceeds the original amount of entanglement
 \item If one had the availability of a perfect rank-1 projecting measurement device, the quantum communication protocols would run with the same fidelity with that device and a state with the amount of entanglement that is present in the measurement space as with the original imperfect device and a state with more entanglement
 \item Increasing or decreasing the entanglement of the original state respectively increases or decreases the maximum possible entanglement one can get in the measurement space, allowing for the unitary operations. This implies that the LOCC maps always decrease the maximum attainable amount of entanglement in the measurement space. Unfortunately this has only been proven for qubits.
\end{enumerate}
The above give us good reason to interpret the entanglement in the measurement space as being the operationally useful entanglement in the context of quantum communication protocols.

We also investigated the way measurement space entanglement can be applied to the mode entanglement. Here we found that mode entanglement is only non-zero in the measurement space when the total number of measurement outcomes is not a prime number. In the case of a completely isolated single particle, we find that due to the existence of only two possible outcomes, the measurement space entanglement always vanishes. However, introducing additional particles that would essentially increase the number of outcomes, makes it possible for the mode entanglement to obtain measurable consequences. \vspace{0.8cm}

\begin{acknowledgments}
 The author would like to thank Dieter Jaksch and Libby Heaney for helpful discussions and the EPSRC research council for funding this research.
\end{acknowledgments}

\bibliographystyle{unsrt}
\bibliography{Entanglement.bib}

\end{document}